\documentclass[journal]{IEEEtran}

\usepackage{enumerate}
\usepackage{cite}
\usepackage{amsmath}
\usepackage{amscd}
\usepackage{amssymb}
\usepackage{latexsym}
\usepackage{array}
\usepackage{tabularx}
\usepackage{epsfig}
\newtheorem{lem}{Lemma}
\newtheorem{ther}{Theorem}

\newtheorem{prop}{Proposition}

\newcommand{\bm}[1]{\mbox{\boldmath{$#1$}}}

\begin{document}
%
\title{On the Proximity Factors of Lattice Reduction-Aided Decoding}
\author{Cong~Ling,~\IEEEmembership{Member,~IEEE}
\thanks{C. Ling is with the Department of Electrical and Electronic Engineering, Imperial College London,
London SW7 2AZ, United Kingdom (e-mail: cling@ieee.org).}}


\maketitle

\begin{abstract}
Lattice reduction-aided decoding features reduced decoding complexity and near-optimum performance
in multi-input multi-output communications. In this paper, a quantitative analysis of lattice
reduction-aided decoding is presented. To this aim, the proximity factors are defined to measure
the worst-case losses in distances relative to closest point search (in an infinite lattice). Upper
bounds on the proximity factors are derived, which are functions of the dimension $n$ of the
lattice alone. The study is then extended to the dual-basis reduction. It is found that the bounds
for dual basis reduction may be smaller. Reasonably good bounds are derived in many cases. The
constant bounds on proximity factors not only imply the same diversity order in fading channels,
but also relate the error probabilities of (infinite) lattice decoding and lattice reduction-aided
decoding.


\end{abstract}



{\center{EDICS: MSP-CODR or SPC-PERF}}

%
\IEEEpeerreviewmaketitle

\section{Introduction}

The linear multi-input multi-output (MIMO) model covers a range of problems in communications, such
as code-division multiple access (CDMA), inter-symbol interference (ISI) channels, linearly
precoded orthogonal frequency-division multiplexing (OFDM), multi-antenna fading channels with or
without linear encoding, multi-antenna broadcast, and cooperative diversity
\cite{Damen,GoldenCode,Hochwald05,Yang07}. For moderate to large problem sizes, (near-)optimum
decoding for MIMO systems represents a challenging problem in communication theory. The naive
exhaustive search for maximum-likelihood (ML) decoding suffers from exponential complexity.
Suboptimum strategies such as zero-forcing (ZF) and successive interference cancelation (SIC)
sometimes incurs heavy performance loss. During the past decades, several improved strategies have
been developed, especially in the context of multiuser detection and equalization. Sequential
decoding \cite{Murugan}, branch-and-bound search \cite{Luojie04}, and semidefinite relaxation
\cite{LuoYu06} represent the state-of-the-art.





The theory of lattices is a powerful, distinctive approach to fast MIMO decoding
\cite{mow:IT,viterbo,agrell}. In digital communications, the signal constellation or codebook is
often drawn from a lattice, especially as the spectrally efficient high-order quadrature amplitude
modulation (QAM) is increasingly used. Lattice decoding exploits the structure of a lattice to
significantly reduce the decoding complexity. Decoding corresponds to the closest vector problem
(CVP) for a lattice. In general, solving the CVP may consist of two stages: lattice reduction and a
local search usually implemented by sphere decoding. While sphere decoding dramatically lowers the
decoding complexity at high signal-to-noise ratio (SNR) \cite{vikaloSD1}, it could be
computationally intensive at low to moderate SNR's. Furthermore, it is known that its complexity
grows exponentially with the system size $n$ for any fixed SNR \cite{jalden}. Thus, to further
widen the decoding speed bottleneck, some performance has to be traded for complexity, i.e., the
CVP needs to be solved approximately. The problem of solving CVP approximately was first addressed
by Babai in \cite{Babai}, which in essence applies ZF or SIC on a reduced lattice. This technique
of approximate lattice decoding is referred to as lattice-reduction-aided decoding in
communications literature \cite{yao,Windpassinger2}, where the expense associated with lattice
reduction can be shared if the channel matrix keeps constant during a frame of data.


\textit{Contributions of this paper:} This paper presents a quantitative understanding of
approximate lattice decoding. More precisely, we shall develop a systematic approach to
characterize its performance in terms of proximity factors, i.e., the worst-case loss in distances
relative to (infinite) lattice decoding, which relates the probabilities of (infinite) lattice
decoding and approximate lattice decoding. This approach is justified by the fact that lattice
decoding often ignores the boundary of finite signals to reduce decoding complexity
\cite{HeshamElGamal04,Jalden:LR-MMSE}. In this paper, the proximity factors are found to be bounded
above by a function of the dimension of the lattice alone. In other words, the output of
approximate lattice decoding is in proximity to that of lattice decoding. As an alternative to
reducing the primal basis, one can reduce the dual basis. This technique has received less
attention thus far. In this paper, we find that in some cases dual basis reduction results in a
smaller upper bound on the proximity factor. The derived bounds serve two purposes: one is to bound
the performance in its own right, the other is to give insights into approximate lattice decoding.
The bounds apply to both fixed and random channels, and hold irrespective of fading statistics.






\textit{Relations to prior works:} A quantitative analysis was first given by Yao and Wornell
\cite{yao}, who showed that the loss is not greater than 3 dB for 2-dimensional complex Gauss
reduction. Our idea of quantifying the performance losses using constant bounds in the general case
started in \cite{ling:jsac} which used the LLL reduction in differential lattice decoding. This
technique has been subsequently developed and improved in a series of conference papers
\cite{LingISTC06, LingISIT06, LingICC08}, eventually leading to the current paper.


Independently, the full receive diversity of Lenstra, Lenstra and Lov\'asz (LLL) reduction in
uncoded MIMO fading channels was shown in \cite{Taherzadeh:IT,XiaoliMa08} by using different
approaches. It was further shown in \cite{Jalden:LR-MMSE} that, remarkably, minimum mean-square
error (MMSE) based lattice-reduction aided decoding achieves the optimal diversity and multiplexing
tradeoff. We would like to point out that the diversity order derived in those papers does not
fully characterize the error performance, as the SNR gap could be arbitrary for the same diversity
order. Therefore, proximity factors offer an approach to performance analysis complementary to
\cite{Taherzadeh:IT,XiaoliMa08,Jalden:LR-MMSE} and could deepen our understanding about approximate
lattice decoding. The derivations in those papers made use of the orthogonality defect of a
lattice, which could also be used to derive upper bounds on proximity factors. However, the bounds
derived in this paper are better.

Our analysis cannot capture the boundary effect associated with lattice decoding. The limitation of
lattice decoding ignoring the boundary was pointed out in \cite{Taherzadeh:Naive}. Yet it was shown
in \cite{HeshamElGamal04, Jalden:LR-MMSE} that the limitation is due to the naive implementation of
lattice decoding. With MMSE preprocessing, the performance loss relative to ML decoding is
insignificant, and in fact lattice decoding achieves the optimum diversity and multiplexing
tradeoff \cite{HeshamElGamal04,Jalden:LR-MMSE}.


%

\textit{Organization:} The rest of the paper is organized as follows. Section II describes the
system model. Section III presents the basics of lattices and lattice reduction that are essential
to the subsequent development of the theory in this paper. In Section IV, we define the proximity
factors of lattice-reduction-aided decoding and show the main results. In Sections V and VI, we
derive the proximity factors for reducing the primal and dual basis, respectively. Section VII
presents numerical results and then ends with a discussion.

\textit{Notation}: Matrices and column vectors are denoted by upper and lowercase boldface letters
(unless otherwise stated), and the transpose, inverse, pseudoinverse of a matrix $\mathbf{B}$ by
$\mathbf{B}^T$, $\mathbf{B}^{-1}$, and $\mathbf{B}^\dagger$, respectively. For convenience, write
$\mathbf{B}^{- T}=(\mathbf{B}^{-1})^T$ and $\mathbf{B}^{\dagger T}=(\mathbf{B}^{\dagger})^T$. The
inner product in the Euclidean space between vectors $\mathbf{u}$ and $\mathbf{v}$ is defined as
$\langle \mathbf{u}, \mathbf{v}\rangle = \mathbf{u}^T\mathbf{v}$, and the Euclidean length
$\|\mathbf{u}\|=\sqrt{\langle \mathbf{u}, \mathbf{u} \rangle}$.



\section{MIMO Decoding}

We consider a real-valued system model. For complex-valued system, an equivalent real system model
with a doubled dimension is used. Reduction of complex lattices without using the equivalent real
model is a topic that has been investigated elsewhere \cite{ComplexLLL}. Let $\mathbf{x} =
(x_1,...,x_n)^T$ be the $n \times 1$ data vector, where each symbol $x_i$ is chosen from a finite
subset of the integer set $\mathbb{Z}$. With proper scaling and shifting \cite{Windpassinger2}, one
has the generic $m \times n$ MIMO system model
\begin{equation}\label{MIMOmodel}
  \mathbf{y}=\mathbf{Bx}+\mathbf{n}
\end{equation}
where $m \geq n$, $\mathbf{y}, \mathbf{n} \in\mathbb{R}^m$ denote the channel output and noise
vectors, respectively, and $\mathbf{B} \in \mathbb{R}^{m \times n}$ is the $m \times n$ full
column-rank matrix. The entries of $\mathbf{n}$ are i.i.d. normal with zero mean and variance
$\sigma^2$. The expression of matrix $\mathbf{B}$ depends on the problem at hand. For example,
$\mathbf{B}$ is the channel matrix for uncoded MIMO systems; $\mathbf{B}$ is composed of the
channel matrix and the generating matrix of the encoding lattice for space-time block codes
\cite{GoldenCode}; $\mathbf{B}$ involves the product-channel matrix for cooperative communications
\cite{Yang07}.

For such a system model, the ML decoder is given by
\begin{equation}\label{ML}
    \hat{\mathbf{x}} = \arg\min_{\mathbf{x} \in \mathcal{C}}\|\mathbf{y}-\mathbf{Bx}\|^2
\end{equation}
where $\mathcal{C}$ stands for the transmitter finite set. Note that the complexity of the standard
ML decoding that uses exhaustive search is exponential in $n$, and also increases with the alphabet
size.

There are linear and nonlinear decoders with cubic complexity. In the linear ZF strategy,
$\mathbf{y}$ is multiplied on the left by the pseudoinverse of $\mathbf{B}$, to yield the detection
rule
\begin{equation}\label{ZF}
  \mathbf{\hat{x}} = \mathcal{Q}\{\mathbf{B}^{\dagger}\mathbf{y}\}
\end{equation}
where $\mathcal{Q}(\cdot)$ denotes the quantization to the nearest integer within the signal
boundary. A well-known drawback of ZF is the effect of noise amplification when the channel matrix
$\mathbf{B}$ is ill-conditioned. By introducing decision feedback in the detection process, the
nonlinear SIC strategy has better performance. One way to do SIC is to perform the QR decomposition
$\mathbf{B} = \mathbf{QR}$, where $\mathbf{Q}$ has orthogonal columns and $\mathbf{R}$ is an upper
triangular matrix \cite{jkzhang}. Multiplying (\ref{MIMOmodel}) on the left with
$\mathbf{Q^\dagger}$ we have
\begin{equation}\label{SICmodel}
  \mathbf{y}' = \mathbf{Q^\dagger y} = \mathbf{Rx}+\mathbf{n}'.
\end{equation}
In SIC, the last symbol $x_n$ is estimated first as $\hat{x}_n = \mathcal{Q}(y'_n/r_{n,n})$. Then
the estimate is substituted to remove the interference term in $y'_{n-1}$ when $x_{n-1}$ is being
estimated. The procedure is continued until the first symbol is detected. That is, we have the
following recursion:
\begin{equation}\label{SIC}
  \hat{x}_i = \mathcal{Q}\left\{\frac{y'_i - \sum_{j=i+1}^{n}r_{i,j}\hat{x}_j}{r_{i,i}}\right\}
\end{equation}
for $i = n, n - 1, ..., 1$.


ZF and SIC may incur heavy performance loss. For example, in single-user uncoded MIMO fading
channels, ZF and SIC are only able to achieve the first-order diversity in an $n \times n$ system
\cite{prasad}, i.e., that of a single antenna, which is far below the order $n$ achieved by ML
decoding. In fact, the diversity order of SIC cannot be increased even with ordering
\cite{Jiang:BLAST07}.



The basic idea behind approximate lattice decoding is to use lattice reduction in conjunction with
traditional low-complexity decoders. With lattice reduction, the basis $\mathbf{B}$ is transformed
into a new basis consisting of roughly orthogonal vectors (this is always possible in a sense
defined later)
\begin{equation}\label{LRaided}
  \mathbf{B'}=\mathbf{BU}
\end{equation}
where $\mathbf{U}$ is a unimodular matrix, i.e., $\mathbf{U}$ contains only integer entries and the
determinant $\det\mathbf{U} = \pm 1$. Indeed, we have the equivalent channel model
\[\mathbf{y} = \mathbf{B'U}^{-1}\mathbf{x}+\mathbf{n} = \mathbf{B'x'}+\mathbf{n}, \quad
\mathbf{x'}=\mathbf{U}^{-1}\mathbf{x}.\] Then conventional decoders (ZF or SIC) are applied on the
reduced basis. For example, ZF obtains the estimation
\begin{equation}\label{LRZF}
  \mathbf{\hat{x}}' = \mathcal{Q}\left\{(\mathbf{B'})^{\dagger}\mathbf{y}\right\}.
\end{equation}
This estimate is then transformed back into $\mathbf{\hat{x}}=\mathbf{U\hat{x}'}$. Since the
equivalent channel is much more likely to be well-conditioned, the effect of noise enhancement will
be moderated. Note that due to the linear transform $\mathbf{U}^{-1}$, it is no longer easy to
control the boundary; thus it is typical to quantize to the nearest integer ignoring the signal
boundary in (\ref{LRZF}).

Both linear and nonlinear detectors can also be designed with respect to the MMSE criterion. The
extension to the MMSE criterion is straightforward by dealing with the augmented channel matrix
\cite{wubbenMMSE}\footnote{Interestingly, this formulation conveniently solves the problem of
lattice reduction in the under-determined case $m<n$.}
\begin{equation}\label{MMSE}
  \mathbf{\mathbf{B}}_\text{a} =
  \begin{bmatrix}
  \mathbf{B}\\ \sigma\mathbf{I}
  \end{bmatrix}.
\end{equation}
When applied to the augmented channel matrix, ZF and SIC are equivalent to the standard MMSE and
MMSE-SIC, respectively. In fact, MMSE is essential for infinite lattice decoding to achieve the
optimum diversity and multiplexing tradeoff for finite constellations
\cite{HeshamElGamal04,Jalden:LR-MMSE}. The proximity factors derived in this paper apply to the
MMSE criterion as well, with the understanding that ZF is replaced by MMSE, and SIC replaced by
MMSE-SIC.

\section{Lattices and Lattice Reduction}

A lattice in the $m$-dimensional embedding Euclidean space $\mathbb{R}^m$ is generated as the
integer linear combination of some set of linearly independent vectors \cite{gruber}:
\begin{equation}\label{Lattice}
    L\triangleq L(\mathbf{B})=\left\{\sum_{i=1}^{n}x_i\mathbf{b}_i\text{ }|\text{ }x_i \in \mathbb{Z}, i = 1, ..., n\right\}
\end{equation}
where $\mathbf{B} = [\mathbf{b}_1, ..., \mathbf{b}_n]$ is called the basis of $L$.

A lattice $L$ can be generated by infinitely many bases, of which one would like to select one that
is in some sense nice or reduced. In many applications, it is advantageous to have the basis
vectors as short as possible. Therefore, lattice reduction, the shortest vector problem (SVP) and
CVP are closely related problems. All bases of the lattice arise by the transformation
$\mathbf{B}'=\mathbf{BU}$, where $\mathbf{U}$ is a unimodular matrix.

The dual lattice $L^*$ of a lattice $L$ is defined as those vectors $\mathbf{u}$, such that the
inner product $\langle\mathbf{u}, \mathbf{v}\rangle \in \mathbb{Z}$, for all $\mathbf{v} \in L$
\cite{gruber}. One might define the dual basis as $\mathbf{B}^{\dagger T}$. In this paper, we
follow the definition of the dual basis $\mathbf{B}^* \triangleq \mathbf{B}^{\dagger T}\mathbf{J}$
in \cite{lagarias}, where
\[
  \mathbf{J}\triangleq
  \begin{bmatrix}
    {0} & {\cdots} & {0} & {1} \\
    {0} & {\cdots} & {1} & {0} \\
    {} & {\ddots} & {} & {} \\
    {1} & {\cdots} & {0} & {0} \
  \end{bmatrix}
\]
is the flipping matrix, i.e., it reverses the columns of $\mathbf{B}^{\dagger T}$ in the left-right
direction. Under this definition, the dual basis satisfy
\begin{equation}\label{reversedual}
  \langle\mathbf{b}_i, \mathbf{b}_k^*\rangle = \delta_{i,n-k+1}.
\end{equation}
where $\mathbf{b}_k^*$ is the $k$-th column of $\mathbf{B}^*$ and $\delta_{i,k}$ is the Kronecker
delta. We have $\det L^* = 1/\det L$ and $L^{**} = L$.

For the sake of convenience, reduction of the dual basis will be referred to as dual reduction.
Dual reduction still results in a reduced basis of the primary lattice. To see this, suppose
$\mathbf{U}^{*}$ is the unimodular matrix arising from reducing the dual basis $\mathbf{B}^*$, then
the corresponding reduced primary basis is given by
\[\mathbf{B'} = (\mathbf{B}^{\dagger T}\mathbf{J}\mathbf{U}^*)^{\dagger T}\mathbf{J} = \mathbf{B}\mathbf{J}(\mathbf{U}^*)^{\dagger T}\mathbf{J}\]
Since $\mathbf{J}(\mathbf{U}^*)^{\dagger T}\mathbf{J}$ is also a unimodular matrix, $\mathbf{B'}$
must be a basis of the primary lattice. In approximate lattice decoding, we can use
$\mathbf{J}(\mathbf{U}^*)^{\dagger T}\mathbf{J}$ as the transformation matrix in (\ref{LRaided}),
whereas anything else remains the same.

Let $\lambda(L)$ be the length of the shortest vector in $L$. The Hermite constant is defined as
\cite{gruber}
\begin{equation}\label{Hermiteconstant}
  \gamma_n \triangleq \sup_L{\frac{\lambda^2(L)}{\det^{2/n} L}},
\end{equation}
where the supremum is taken over all lattices $L$ of dimension $n$. Hermite showed that $\gamma_n$
is bounded \cite{gruber}. Mordell derived the inequality $\gamma_n^{n-2} \leq \gamma_{n-1}^{n-1}$
\cite{Mordell}.
The upper bound was improved by Blichfeldt \cite{Blich}
\begin{equation}\label{Hermiteestimp}
    \gamma_n \leq \frac{2}{\pi}\Gamma^{2/n}\left(2+\frac{n}{2}\right),
\end{equation}
which has an asymptotic value $n/\pi e$. Even better estimates are known as $n$ goes to infinity
\cite{lagarias}:
\begin{equation}\label{Hermite3}
    \frac{n}{2\pi e}+o(n) \leq \gamma_n \leq \frac{0.872n}{\pi e}+o(n), \quad n\rightarrow\infty.
\end{equation}
Moreover, the Hermite constant is known exactly for $n \leq 8$:
\begin{equation}\label{Hermite8}
\begin{split}
\gamma_1 &= 1, \gamma_2^2 = 4/3, \text{ }\gamma_3^3 = 2, \text{ }\text{ }\gamma_4^4 = 4,\\
\gamma_5^5 &= 8, \gamma_6^6 = 64/3, \gamma_7^7 = 64, \gamma_8^8 = 2^8.
\end{split}
\end{equation}



The theory of lattice reduction \cite{gruber} started with Gauss who introduced the concept of a
lattice and found short vectors in 2-dimensional lattices. Reduction for general $n$ was proposed
by Hermite. More manageable versions were given by Minkowski \cite{BK:Minkowski}, and Korkin and
Zolotarev (KZ) \cite{KZ}, which try to find short basis vectors recursively. For a
Minkowski-reduced basis $\mathbf{B}$, $\mathbf{b}_1$ is a shortest nonzero vector in $L$, and
$\mathbf{b}_i$ for $i=2,\cdots,n$ is the shortest vector in $L$ such that $[\mathbf{b}_1,\cdots
\mathbf{b}_i]$ may be extended to a basis of $L$. In KZ reduction, this is performed in the linear
space orthogonal to the vectors already found. In 1982, Lenstra, Lenstra and Lov\'asz proposed the
first polynomial time algorithm which finds a vector not longer than $2^{(n-1)/2}$ times the
shortest nonzero vector.

Both KZ and LLL reduction can be seen as possible generalizations of the Gauss reduction, and they
can be defined in terms of the Gram-Schmidt orthogonalization. One can compute the Gram-Schmidt
orthogonalization by the recursion \cite{golub}
\begin{equation}\label{GSO}
\begin{split}
  \mathbf{\hat{b}}_i &= \mathbf{b}_i - \sum_{j=1}^{i-1}{\mu_{i,j}\mathbf{\hat{b}}_j},\quad \text{for } i = 1,...,n
\end{split}
\end{equation}
where $\mu_{i,j} = \langle \mathbf{b}_i, \mathbf{\hat{b}}_j\rangle/\|\mathbf{\hat{b}}_j\|^2$. Note
that the Gram-Schmidt orthogonalization depends on the order of the basis vectors. In matrix
notation, Gram-Schmidt orthogonalization can be written as $\mathbf{B}=\mathbf{\hat{B}}{\bm
\mu}^T$, where $\mathbf{\hat{B}}=[\mathbf{\hat{b}}_1, ..., \mathbf{\hat{b}}_n]$, and $\bm \mu$ is a
lower-triangular matrix with unit diagonal elements. In relation to the QR decomposition, one has
$\mu_{j,i}=r_{i,j}/r_{i,i}$ and $\mathbf{\hat{b}}_i=r_{i,i}\cdot \mathbf{q}_i$ where $\mathbf{q}_i$
is the $i$-th column of $\mathbf{Q}$.



\subsection{KZ Reduction}
A basis $\mathbf{B}$ is said to be KZ-reduced if it satisfies the following conditions
\cite{lagarias}:

\begin{itemize}
  \item $\mathbf{b}_1$ is the shortest nonzero vector of $L$;
  \item $|\mu_{i,j}| \leq 1/2$, for $1 \leq j < i \leq n$;
  \item If $L^{(n-1)}$ denotes that orthogonal projection of $L$ on the orthogonal complement of $\mathbb{R}\mathbf{b}_1$,
   then the projections $\mathbf{b}_i - \mu_{i,1}\mathbf{b}_1$ of $\mathbf{b}_2, ..., \mathbf{b}_n$ yield a KZ basis
   of $L^{(n-1)}$. Equivalently, $\mathbf{\hat{b}}_i$ is a shortest vector of $L^{(n-i+1)}$.
\end{itemize}

Finding a KZ basis is polynomial-time equivalent to finding the shortest vector. Hence, the KZ
basis is not easy to find for lattices of high dimension. When $n=2$, KZ reduction is the same as
Gaussian reduction.


\subsection{LLL Reduction}
A basis is LLL reduced if \cite{LLL}
\begin{itemize}
  \item $|\mu_{i,j}| \leq 1/2$, for $1 \leq j < i \leq n$;
  \item $\|\mathbf{\hat{b}}_i\|^2 + \mu_{i,i-1}^2\|\mathbf{\hat{b}}_{i-1}\|^2 \geq \delta
  \|\mathbf{\hat{b}}_{i-1}\|^2$, for $1<i\leq n$.
\end{itemize}

The first clause is called size reduction, while the second is known as the Lov\'asz condition. The
parameter $\delta$ takes values in the interval $(1/4, 1]$ and it is common to choose $\delta =
3/4$ for a tradeoff between quality and complexity. It is well known that the standard LLL
algorithm requires $O(n^4)$ arithmetic operations for $n\times n$ integer bases \cite{LLL}.
Recently, we established the $O(n^3 \log n)$ average time bound for real-valued bases whose vectors
are i.i.d. standard normal \cite{LingISIT07}, which is lower than thought before. A similar result
on the average number of iterations was obtained in \cite{Jalden08}.

Define $\beta \triangleq 1/(\delta - 1/4)$. An LLL-reduced lattice has the following properties:
\begin{itemize}
  \item $\|\mathbf{b}_1\| \leq \beta^{(n-1)/2}\lambda(L)$;
  \item $\|\mathbf{b}_1\| \leq \beta^{(n-1)/4}\det^{1/n}L$;
  \item $\prod_{i=1}^n{\|\mathbf{b}_i\|} \leq \beta^{n(n-1)/4}\det L$.
\end{itemize}
These inequalities indicate in various senses that the basis vectors are short. These bounds are
tight in the worst case in the sense that there exist bases reaching the bounds.

When $n=2$ and $\delta=1$, LLL reduction coincides with Gaussian reduction.

\section{Proximity Factors}
In this section, we shall introduce an analytic tool for approximate lattice decoding. The lattice
literature is usually concerned with the approximation factor of CVP. For instance, Babai proved
that ZF and SIC on an LLL-reduced basis find the closest vector, up to a factor $1+2n(9/2)^{n/2}$
and $2^{n/2}$ (for $\delta = 3/4$), respectively [7]. However, such results do not directly
translate into how close approximate lattice decoding is to (infinite) lattice decoding in terms of
the minimum distance, which is more useful in digital communications. To determine the proximity,
we examine the decision regions of the decoders and the associated distances.

\subsection{Definition}
\begin{figure}[t]

\vspace{-0.5cm}

\centering\centerline{\epsfig{figure=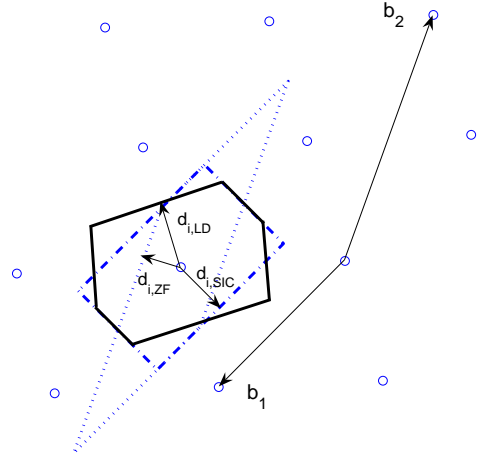,width=15cm}}

\vspace{-1.0cm}

\caption{Decision regions of ZF (dotted), SIC (dash-dotted) and (infinite) lattice decoding (solid)
and the corresponding distances in a 2-dimensional lattice.}

\vspace{-0.5cm}

\label{fig:decisionregion}
\end{figure}

We consider a fixed but arbitrary lattice. Since we assume an infinite lattice, the signal space is
geometrically uniform. Without loss of generality, let the transmitted lattice point be $\mathbf{x}
= \mathbf{0}$. The decision regions $\mathcal{R}$ of various strategies are depicted in Fig.
\ref{fig:decisionregion} for a 2-dimensional lattice.

The Voronoi cell, corresponding to the decision region of (infinite) lattice decoding, is a
polytope defined by
\[\mathcal{R}_{\text{LD}} = \{\mathbf{y}: \|\mathbf{y} - \mathbf{v}\| \geq \|\mathbf{y}\|, \quad \forall \mathbf{v} \in L\},\]
where the subscript LD stands for lattice decoding. Each facet of the Voronoi cell lies on the
perpendicular bisector of the line connecting the origin and the Voronoi neighbor $\mathbf{v}_i \in
\mathcal{N}$, where $\mathcal{N}$ represents the Voronoi neighborhood.

Let $d_{i,\text{LD}}$ be the Euclidean distance from point $\mathbf{0}$ to the $i$-th facet of
$\mathcal{R}_{\text{LD}}$. It is easy to see from Fig. \ref{fig:decisionregion} that
$d_{i,\text{LD}}=\|\mathbf{v}_i\|/2$. The decision regions of both ZF and SIC are polyhedra with
$2n$ facets and are symmetric with respect to the origin. Specifically, the ZF decision region is
just the fundamental parallelogram $\mathcal{P}(\mathbf{B})=\{\mathbf{{B}}\mathbf{a}, |a_i| \leq
1/2\}$ centered at $\mathbf{x} = \mathbf{0}$. Meanwhile, the SIC decision region is a rectangle
specified by the Gram-Schmidt vectors as $\{\mathbf{\hat{B}}\mathbf{a}, |a_i| \leq 1/2\}$
\cite{Verdu,yao,ling:jsac}.

Let $d_{i,\text{ZF}}$ and $d_{i,\text{SIC}}$ be the Euclidean distance from point $\mathbf{0}$ to
the $i$-th facet of the decision regions of ZF and SIC, respectively. We are concerned with a
spectrum of $n$ distances due to symmetry. It can be seen from Fig. \ref{fig:decisionregion} that
the distance spectrum of ZF is given by
\[d_{i,\text{ZF}} = \frac{1}{2}\|\mathbf{b}_i\|\sin \theta_i, \quad i = 1, ..., n\]
where $\theta_i$ denotes the angle between $\mathbf{b}_i$ and the linear space spanned by the other
$n - 1$ basis vectors. In Appendix \ref{appendix:theta}, we derive a formula to calculate
$\theta_i$ from $\mathbf{B}$ and its Gram-Schmidt orthogonalization, which will be used in latter
sections to bound $d_{i,\text{ZF}}$.

For SIC, the distance spectrum is given by
\[d_{i,\text{SIC}} = \frac{1}{2}\|\mathbf{\hat{b}}_i\| = \frac{1}{2}\|\mathbf{b}_i\|\sin\phi_i, \quad i = 1, ..., n\]
where $\phi_i$ is be the angle between $\mathbf{b}_i$ and the hyperplane spanned by $\mathbf{b}_1,
..., \mathbf{b}_{i-1}$. Note that $\|\mathbf{b}_i\|\sin \theta_i$ is the distance between the
vector point $\mathbf{b}_i$ and the linear space spanned by the other $n - 1$ basis vectors. Since
$\mathbf{\hat{b}}_i$ only needs to be orthogonal to $\mathbf{b}_1, ..., \mathbf{b}_{i-1}$, we must
have $\theta_i \leq \phi_i$ and hence $d_{i,\text{ZF}} \leq d_{i,\text{SIC}}$.

The minimum distance $d_{\text{LD}} \triangleq \min_i\{d_{i,\text{LD}}\} = \lambda(L)/2$ plays an
important role. We are motivated to define the {\it proximity factors} measuring the proximity
between the performance of (infinite) lattice decoding and approximate lattice decoding as follows:
\begin{equation}\label{PF}
\begin{split}
  \rho_{i,\text{ZF}} & \triangleq \sup_{\mathbf{B} \in \mathcal{B}_{\text{Reduced}}}{ \frac{d^2_{\text{LD}}}{d^2_{i,\text{ZF}}}} = \sup_{\mathbf{B} \in \mathcal{B}_{\text{Reduced}}}{ \frac{\lambda^2(L)}{\|\mathbf{b}_i\|^2\sin^2\theta_i}}\\
  \rho_{i,\text{SIC}} &\triangleq \sup_{\mathbf{B} \in \mathcal{B}_{\text{Reduced}}}{ \frac{d^2_{\text{LD}}}{d^2_{i,\text{SIC}}}} = \sup_{\mathbf{B} \in \mathcal{B}_{\text{Reduced}}}{
  \frac{\lambda^2(L)}{\|\mathbf{\hat{b}}_i\|^2}}\\
\end{split}
\end{equation}
where the supremum is taken over the set $\mathcal{B}_{\text{Reduced}}$ of basis matrices
satisfying a certain {\it reduction criterion} for any $n$-dimensional lattices $L$. Furthermore,
we define $\rho_{\text{ZF}} \triangleq \max_{1\leq i \leq n}{\rho_{i,\text{ZF}}}$ and
$\rho_{\text{SIC}} \triangleq \max_{1\leq i \leq n}{\rho_{i,\text{SIC}}}$, which quantify the
worst-case loss of approximate lattice decoding in the minimum squared Euclidean distance. It is
noteworthy that the proximity factors are not a function of any particular lattice $L$, and that
without reduction they are unbounded. However, we shall show in the latter Sections that with
lattice reduction the proximity factor is bounded by a constant that is a function of $n$ only.

\subsection{Distance Properties of Dual Reduction}



By (\ref{reversedual}), $\mathbf{b}_{n-i+1}^*$ of the dual basis is perpendicular to the
$(n-1)$-dimensional hyperplane spanned by vectors $\mathbf{b}_1, ..., \mathbf{b}_{i-1}$,
$\mathbf{b}_{i+1}, ...,\mathbf{b}_n$. Furthermore, since $(\pi/2 - \theta_i)$ is the angle between
$\mathbf{b}_i$ and $\mathbf{b}_{n-i+1}^*$, we have
\[\langle\mathbf{b}_i, \mathbf{b}_{n-i+1}^*\rangle
=\|\mathbf{b}_{n-i+1}^*\|\cdot\|\mathbf{b}_i\|\cdot\cos(\pi/2-\theta_i)=
\|\mathbf{b}_{n-i+1}^*\|\cdot\|\mathbf{b}_i\|\cdot\sin\theta_i=1.\] Therefore, the following
relation holds for dual reduction:
\begin{equation}\label{ZFlength}
  d_{i,\text{ZF}}=\frac{1}{2} \|\mathbf{b}_i\|\sin \theta_i = \frac{1}{2} \|\mathbf{b}_{n-i+1}^*\|^{-1}.
\end{equation}
It indicates that we can have large distances in ZF if the vectors of the dual basis are short.

\begin{prop}\label{prop1}
Let $\mathbf{B}=\mathbf{\hat{B}}\bm{\mu}^T$ and $\mathbf{B}^*=\mathbf{\hat{B}}^*(\bm{\mu}^*)^T$ be
the Gram-Schmidt orthogonalization of the primal basis $\mathbf{B}$ and dual basis $\mathbf{B}^*$,
respectively. Then
\begin{equation}\label{dual}
\begin{split}
    \mathbf{\hat{B}}^* &= \mathbf{\hat{B}}^{\dagger T}\mathbf{J},\\
    \bm{\mu}^* &= \mathbf{J}\bm{\mu}^{-T}\mathbf{J}.
\end{split}
\end{equation}
\end{prop}

\begin{proof}
By definition, we have that for matrix $\mathbf{B}$ with full column rank \[\mathbf{B}^{\dagger T}
= [(\mathbf{B}^T\mathbf{B})^{-1}\mathbf{B}^T]^T = \mathbf{B}(\mathbf{B}^T\mathbf{B})^{-T}.\]
Substituting $\mathbf{B}=\mathbf{\hat{B}}\bm{\mu}^T$, we arrive at \[\mathbf{B}^{\dagger T} =
\mathbf{\hat{B}}\bm{\mu}^T (\bm{\mu} \mathbf{\hat{B}}^T\mathbf{\hat{B}}\bm{\mu}^T)^{-T} =
\mathbf{\hat{B}} ( \mathbf{\hat{B}}^T\mathbf{\hat{B}})^{-T} \bm{\mu}^{-1} =
\mathbf{\hat{B}}^{\dagger T}\bm{\mu}^{-1}.\] Hence the reversed dual basis can be written as
\[\mathbf{B}^*=\mathbf{B}^{\dagger T}\mathbf{J} = \mathbf{\hat{B}}^{\dagger
T}\bm{\mu}^{-1}\mathbf{J}.\] As $\mathbf{J}\mathbf{J}=\mathbf{I}$, it can be rewritten into
\[\mathbf{B}^*= \mathbf{\hat{B}}^{\dagger
T}\mathbf{J}\cdot\mathbf{J}\bm{\mu}^{-1}\mathbf{J}.\] Since $\mathbf{J}\bm{\mu}^{-1}\mathbf{J}$ is
upper triangular with unit diagonal elements, and since the columns of $\mathbf{\hat{B}}^{\dagger
T}\mathbf{J}$ are orthogonal to each other, this is exactly the Gram-Schmidt orthogonalization of
$\mathbf{B}^*$. Therefore, we have $\mathbf{\hat{B}}^* = \mathbf{\hat{B}}^{\dagger T}\mathbf{J}$
and $(\bm{\mu}^*)^T = \mathbf{J}\bm{\mu}^{-1}\mathbf{J}$. As $\mathbf{J}=\mathbf{J}^T$, $\bm{\mu}^*
= \mathbf{J}\bm{\mu}^{-T}\mathbf{J}$ holds.
\end{proof}

$\mathbf{\hat{B}}^* = \mathbf{\hat{B}}^{\dagger T}\mathbf{J}$ implies
\begin{equation}\label{reverseGS1}
  \hat{\mathbf{b}}_{n-i+1}^*=\frac{\hat{\mathbf{b}}_{i}}{\|\hat{\mathbf{b}}_{i}\|^2},
  \quad i = 1,2,...,n
\end{equation}
where $\hat{\mathbf{b}}_1^*, ..., \hat{\mathbf{b}}_n^*$ are the Gram-Schmidt vectors of the dual
basis $\mathbf{B}^*$. The elegant relation (\ref{reverseGS1}) was derived earlier in
\cite{lagarias, Howgrave-Graham97}. The relation $\bm{\mu}^* = \mathbf{J}\bm{\mu}^{-T}\mathbf{J}$
is new; it corresponds to flipping $\bm{\mu}^{-1}$ with respect to the anti-diagonal.

Meanwhile, by (\ref{reverseGS1}),
\begin{equation}\label{reverseGS2}
  d_{i,\text{SIC}}=\frac{1}{2}\|\hat{\mathbf{b}}_{i}\|=\frac{1}{2}\|\hat{\mathbf{b}}_{n-i+1}^*\|^{-1}.
\end{equation}
It shows that SIC can have large distances if the reversed dual basis has short Gram-Schmidt
vectors.

\subsection{Performance Bounds for Approximate Lattice Decoding}
For each decoder, an error occurs when the noise falls outside of $\mathcal{R}$. Accordingly, given
the basis $\mathbf{B}$, the error probability for vector $\mathbf{x}$ is given by
\begin{equation}\label{SEP}
  P_e(\mathbf{B}) = P(\mathbf{\hat{x}}\neq \mathbf{0}|\mathbf{x}=\mathbf{0}) = P(\mathbf{n} \in \overline{\mathcal{R}}).
\end{equation}

To keep the results general, we write $\text{SNR}=c/\sigma^2$, where $c$ is a constant depending on
the problem. By the symmetry of the Voronoi cell (i.e., there are at least two vectors with the
shortest length,) we have the lower bound on the conditional decoding error probability of
(infinite) lattice decoding
\begin{equation}\label{SEPMLLB}
  P_{e,\text{LD}}(\text{SNR},\mathbf{B}) \geq 2{Q\left(\frac{d_{\text{LD}}}{\sigma}\right)} = 2{Q\left(\sqrt{\frac{d_{\text{LD}}^2\cdot\text{SNR}}{c}}\right)}.
\end{equation}

Meanwhile, the union bound on the conditional error probability of ZF reads
\begin{equation}\label{SEPZF}
  P_{e,\text{ZF}}(\text{SNR},\mathbf{B}) \leq 2\sum_{i=1}^{n}{Q\left(\frac{d_{i,\text{ZF}}}{\sigma}\right)}
\end{equation}
where the factor 2 is due to symmetry. The union bound for SIC admits a form similar to
(\ref{SEPZF}). Given the same basis matrix $\mathbf{B}$, the conditional error probability of
lattice reduction-aided ZF can be bounded above as
\begin{equation}\label{ZFB}
  P_{e,\text{ZF}}(\text{SNR},\mathbf{B}) \leq
  2\sum_{i=1}^{n}{Q\left(\frac{d_{\text{LD}}}{\sqrt{\rho_{i,\text{ZF}}}\sigma}\right)}=2\sum_{i=1}^{n}{Q\left(\sqrt{\frac{d_{\text{LD}}^2\cdot\text{SNR}}{c\cdot\rho_{i,\text{ZF}}}}\right)}.
\end{equation}
since $d^2_{i,\text{ZF}} \geq \rho_{i,\text{ZF}}\cdot d^2_{\text{LD}}$ by definition (\ref{PF}) and
since $Q(\cdot)$ is a decreasing function. It is worth pointing out that while the distance
$d_{\text{LD}}$ is a function of $\mathbf{B}$, $\rho_{i,\text{ZF}}$ is not.

Now, combining (\ref{SEPMLLB}) and (\ref{ZFB}), we have
\begin{equation}\label{SEP_ExactB}
  P_{e,\text{ZF}}(\text{SNR},\mathbf{B}) \leq
  \sum_{i=1}^{n}{P_{e,\text{LD}}\left(\frac{\text{SNR}}{\rho_{i,\text{ZF}}},\mathbf{B}\right)}.
\end{equation}
Since (\ref{SEP_ExactB}) holds for any $\mathbf{B}$, averaging out $\mathbf{B}$ we obtain
\begin{equation}\label{SEP_Exact}
  P_{e,\text{ZF}}(\text{SNR}) \leq
  \sum_{i=1}^{n}{P_{e,\text{LD}}\left(\frac{\text{SNR}}{\rho_{i,\text{ZF}}}\right)}
\end{equation}
for arbitrary SNR. In particular,
\begin{equation}\label{SEP_Exact2}
  P_{e,\text{ZF}}(\text{SNR}) \leq nP_{e,\text{LD}}\left(\frac{\text{SNR}}{\rho_{\text{ZF}}}\right).
\end{equation}


The relations (\ref{SEP_Exact}) and (\ref{SEP_Exact2}) hold irrespective of fading statistics, and
similar relations exist for SIC. They reveal, in a quantitative manner, that approximate lattice
decoding performs within a constant bound from lattice decoding. The mere effect on the error rate
curve is a shift from that of lattice decoding, up to a multiplicative factor $n$, which obviously
does not change the diversity order. In other words, the diversity order is the same as that of
lattice decoding. Therefore, existing results on the diversity order of lattice decoding can be
extended to approximate lattice decoding. In particular, since lattice decoding with MMSE
pre-processing achieves the optimal diversity-multiplexing tradeoff for approximately universal
codes \cite{Jalden:LR-MMSE}, approximate lattice decoding with MMSE pre-precessing also does.
Moreover, since lattice decoding achieves full receive diversity in the uncoded V-BLAST system
\cite{Taherzadeh:IT}, approximate lattice decoding also achieves full diversity. This provides an
alternative way of showing the diversity order of lattice-reduction aided decoding given in
\cite{Taherzadeh:IT, Jalden:LR-MMSE}.



\subsection{Main Results}

Table \ref{tab:comparison} shows the main results of this paper, i.e., the bounds on proximity
factors for reducing the primary basis and for reducing the dual basis. The derivations will be
given in the following two sections. Notice that the value $4/3$ is exact for Gaussian reduction,
and could be obtained by specializing the bounds for LLL or KZ reduction. In the meantime, the
bounds for LLL and KZ reduction in Table \ref{tab:comparison} are less explicit, i.e., they are
expressed in the parameter $\beta$ for LLL reduction and in the Hermite and KZ constants for KZ
reduction. To give the readers a feeling about the dependence of the bounds on $n$, Table
\ref{tab:comparison2} shows more explicit bounds for LLL and KZ reduction, where the bounds are
however less tight. For LLL reduction, the bounds on proximity factors are exponential; if
expressed in dB, they are linear with $n$. It is seen that reducing the dual basis leads to smaller
bounds in many cases. Namely, when the dual basis is reduced, LLL-ZF has a smaller base in its
exponential bound, while KZ reduction results in polynomial bounds on proximity factors for both ZF
and SIC. It is also interesting that, unlike primal basis reduction, LLL-ZF and LLL-SIC have the
same bounds on proximity factors when the dual basis is reduced.

\begin{table*}
\renewcommand{\arraystretch}{1.5}
\caption{Upper Bounds on Proximity Factors ($n$: dimension; $\beta=1/(\delta-1/4), 1/4<\delta\leq
1$; $\gamma_j$: Hermite constant; $\xi_j$: KZ constant)} \label{tab:comparison} \centering
\begin{tabular}{|c||c|c|c|c|c|c|}
\hline
       &Gauss-ZF &Gauss-SIC & LLL-ZF & LLL-SIC & KZ-ZF & KZ-SIC \\
\hline Primal Reduction&$\frac{4}{3}$ &$\frac{4}{3}$ &
$\frac{\beta}{9\beta-4}\left(\frac{9\beta}{4}\right)^{n-1} + \frac{8\beta-4}{9\beta-4}$ & $1+
\frac{\beta^n-\beta}{4(\beta-1)}$ & $1+\sum_{j=1}^{n-1}{\left(\frac{1}{3}\right)^2
  \left(\frac{3}{2}\right)^{2j}\xi_{j+1}}$ & $\xi_n$ \\
\hline Dual Reduction &$\frac{4}{3}$ &$\frac{4}{3}$ & $\beta^{n-1}$ & $\beta^{n-1}$ & $1+\frac{1}{4}\sum_{j=2}^{n}\gamma_{j}^2$ & $\max_{1\leq j\leq n}{\gamma_j^2}$\\
\hline
\end{tabular}
\end{table*}

\begin{table}
\renewcommand{\arraystretch}{1.5}
\caption{More Explicit But Less Tight Upper Bounds on Proximity Factors ($n$: dimension;
$\delta=3/4$ for LLL reduction)} \label{tab:comparison2} \centering
\begin{tabular}{|c||c|c|c|c|}
\hline
       &LLL-ZF & LLL-SIC & KZ-ZF & KZ-SIC \\
\hline Primal Reduction&$\left(\frac{9}{2}\right)^n$ & $2^{n}$ &
$\left(\frac{9}{4}\right)^{n} \cdot n^{1+\ln n}$ & $n^{1+\ln n}$ \\
\hline Dual Reduction &$2^{n}$ & $2^{n}$ & $n^3$ & $n^2$\\
\hline
\end{tabular}
\end{table}

\section{Primal Basis Reduction}\label{Sect:Primal}
With lattice reduction, the supremum in (\ref{PF}) is taken over the set of reduced bases of the
lattice $L$. Consequently, the proximity factors will be bounded. In this Section, we shall derive
explicit upper bounds for LLL and KZ reductions. The bounds are expressed in closed form and turn
out to be functions of $n$ only.

%
%
%

\subsection{LLL Reduction}
The derivation for LLL reduction will be done by adapting the techniques of the original LLL paper
\cite{LLL} and Babai \cite{Babai}. While the exponential bounds in \cite{LLL, Babai} have the best
base, in this subsection we will improve the bounds by a constant factor.

\subsubsection{SIC}
\begin{ther}\label{theorem:LLL-SIC} For LLL-SIC, the proximity factors are bounded by
\begin{equation}\label{}
\begin{split}
   \rho_{i,\text{SIC}} \leq 1 + \frac{\beta^i-\beta}{4(\beta-1)}.
\end{split}
\end{equation}
\end{ther}

\begin{proof}
From \cite{LLL} one has
\begin{equation}\label{GSlength}
  \|\mathbf{\hat{b}}_j\|^2 \leq \beta^{i-j}\|\mathbf{\hat{b}}_i\|^2, \quad \text{for}\: \:1\leq j < i \leq
  n.
\end{equation}
Substituting this into the representation of $\mathbf{b}_i$ in Gram-Schmidt vectors
\begin{equation}\label{}
  \|\mathbf{b}_i\|^2 = \|\mathbf{\hat{b}}_i\|^2 +
  \sum_{j=1}^{i-1}{\mu_{i,j}^2\|\mathbf{\hat{b}}_j\|^2},
\end{equation}
we obtain
\begin{equation}\label{vectorlength}
\begin{split}
  \|\mathbf{b}_i\|^2 &\leq \left(1 + \sum_{j=1}^{i-1}{\beta^{i-j}/4}\right)\|\mathbf{\hat{b}}_i\|^2.\\
\end{split}
\end{equation}
Note that (\ref{vectorlength}) was also given in \cite{LLL} for the special case $\delta=3/4$.



Let $\mathcal{B}_{\text{LLL}}$ be the set of LLL-reduced basis matrices. Since $\lambda(L) \leq
\|\mathbf{b}_i\|$, the loss in the $i$-th squared Euclidean distance is bounded by
\begin{equation}\label{betterbound}
\begin{split}
   \rho_{i,\text{SIC}} &= \sup_{\mathbf{B}\in \mathcal{B}_{\text{LLL}}}{ \frac{\lambda^2(L)}{\|\mathbf{\hat{b}}_i\|^2}} \\
   & \leq  \left(1 + \frac{\beta^i-\beta}{4(\beta-1)}\right)\sup_{\mathbf{B}\in \mathcal{B}_{\text{LLL}}}{\frac{\lambda^2(L)}{\|\mathbf{b}_i\|^2}}\\
   &\leq 1 + \frac{\beta^i-\beta}{4(\beta-1)}.
\end{split}
\end{equation}
\end{proof}

It is easy to verify that when $\beta \geq 4/3$ (which is the case for LLL reduction)
(\ref{betterbound}) is smaller than the bound
\begin{equation}\label{expbound}
    \rho_{i,\text{SIC}} \leq \beta^{i-1}
\end{equation}
directly obtainable from \cite{LLL}. Obviously,
\begin{equation}\label{rhoSICLLL}
  \rho_{\text{SIC}} \leq 1+ \frac{\beta^n-\beta}{4(\beta-1)} \leq \beta^{n-1}.
\end{equation}

\subsubsection{ZF}

The derivation for LLL-ZF needs the following lemma. It improves Babai's lower bound on $\theta_i$
\cite{Babai} and extends it to arbitrarily values of $\delta \in (1/4, 1)$. As recognized by Babai
\cite{Babai}, the lower bound on $\theta_i$ describes a geometric feature of an LLL-reduced basis,
which is of independent interest. The proof is given in Appendix \ref{appendix:Newgamma}.

\begin{lem}\label{lem:lemma1}
If $\mathbf{B}$ is an LLL-reduced basis of lattice $L(\mathbf{B})$, then
\begin{equation}\label{lemma1}
\begin{split}
\sin^2 \theta_i \geq &\left[\frac{\beta}{9\beta-4}\left(\frac{9\beta}{4}\right)^{n-i} +
\frac{8\beta-4}{9\beta-4}\right]^{-1} \times\\
&\left[1+ \frac{\beta^i-\beta}{4(\beta-1)}\right]^{-1}
\end{split}
\end{equation}
for $1 \leq i \leq n$.
\end{lem}

It is worth mentioning that the lower bound of Lemma~\ref{lem:lemma1} is tight. The equality is
achieved when $\|\mathbf{\hat{b}}_j\|^2 = \beta^{i-j}\|\mathbf{\hat{b}}_i\|^2$ and $\mu_{i,j} =
-1/2$ for all $1\leq j < i \leq n$. This can be seen from the proof.

Applying Lemma~\ref{lem:lemma1} and the trivial inequality $\|\mathbf{b}_i\| \geq \lambda(L)$ to
\begin{equation}\label{PFZF2}
\begin{split}
  \rho_{i,\text{ZF}} &\leq  \sup_{\mathbf{B}\in \mathcal{B}_{\text{LLL}}}{\frac{\lambda^2(L)}{\|\mathbf{b}_i\|^2}\frac{1}{\sin^2\theta_i}},\\
\end{split}
\end{equation}
we obtain the following theorem:
\begin{ther}\label{theorem:LLL-ZF} For LLL-ZF, the proximity factors are bounded by
\begin{equation}\label{PFZF}
\begin{split}
  \rho_{i,\text{ZF}} &\leq  \left[\frac{\beta}{9\beta-4}\left(\frac{9\beta}{4}\right)^{n-i} + \frac{8\beta-4}{9\beta-4}\right] \left(1+\frac{1}{4} \cdot
\frac{\beta^i-\beta}{\beta-1}\right).
\end{split}
\end{equation}
\end{ther}

\vspace{0.5cm}

Since the upper bound is maximized when $i = 1$, we have
\begin{equation}\label{rhoZFLLL}
  \rho_{\text{ZF}} \leq \frac{\beta}{9\beta-4}\left(\frac{9\beta}{4}\right)^{n-1} + \frac{8\beta-4}{9\beta-4}.
\end{equation}
This is better than the bound $\rho_{\text{ZF}} \leq \left(\frac{9\beta}{4}\right)^{n-1}$ obtained
previously in \cite{LingISTC06}.

\subsection{KZ Reduction}

\subsubsection{SIC}
To derive the proximity factors for SIC with KZ reduction, we make use of the KZ constant defined
in \cite{schnorr87} as
\begin{equation}\label{KZconstant}
  \xi_n \triangleq \sup_{\mathbf{B}\in \mathcal{B}_{\text{KZ}}}{\frac{\|\mathbf{b}_{1}\|^2}{\|\mathbf{\hat{b}}_{n}\|^2}}
\end{equation}
where $\mathcal{B}_{\text{KZ}}$ is the set of all KZ-reduced bases of dimension $n$. The KZ
constant is bounded in terms of the Hermite constant by \cite{schnorr87}
\begin{equation}\label{KZconstantBound}
  \xi_n \leq \gamma_n \gamma_{n}^{1/(n-1)} \gamma_{n-1}^{1/(n-2)}...\gamma_2
\end{equation}
for $n \geq 2$. Since $\|\mathbf{b}_{1}\| = \lambda(L)$ in a KZ-reduced lattice, we have
\[\rho_{n,SIC} = \xi_n.\]
To obtain $\rho_{i,\text{SIC}}$ for $i < n$, note that if $\mathbf{B}$ is a KZ-reduced basis of
lattice $L(\mathbf{B})$, then $\mathbf{B}_i = [\mathbf{b}_1, ..., \mathbf{b}_i]$ is a KZ reduced
basis of the sublattice $L_i = L(\mathbf{B}_i)$. This follows from the definition of KZ reduction
and the fact that $L_i \subseteq L$. Accordingly, we have
\begin{ther}\label{theorem:KZ-SIC} For KZ-SIC, the proximity factors are bounded by
\[\rho_{i,\text{SIC}} = \xi_i.\]
\end{ther}
\vspace{0.5cm}

Since the KZ constant $\xi_n$ is a nondecreasing function of $n$ \cite{schnorr87}, we have
$\rho_{\text{SIC}} = \xi_n$, i.e., the KZ constant exactly quantifies the worst-case performance
loss of SIC. The KZ constant is exactly known for $n = 2$ and $n=3$: $\xi_2 = 4/3$ and $\xi_3 =
3/2$ \cite{schnorr87}. It is also known that \cite{schnorr87,lagarias}
\begin{equation}\label{KZconstantBound2}
  \xi_n \leq n^{1+\ln n},
\end{equation}
which is sub-exponential.

\subsubsection{ZF}
\begin{ther} For KZ-ZF, the proximity factors are bounded by
\begin{equation}\label{rhoiZFKZ2}
\begin{split}
  \rho_{i,\text{ZF}} \leq \xi_{i}+\sum_{j=1}^{n-i}{\left(\frac{1}{3}\right)^2 \left(\frac{3}{2}\right)^{2j}\xi_{i+j}}.\\
\end{split}
\end{equation}
\end{ther}

\begin{proof}
By (\ref{BabaiAngle}), we have
\begin{equation}\label{rhoiZF-anybasis}
\begin{split}
  \rho_{i,\text{ZF}} &= \sup_{\mathbf{B}\in \mathcal{B}_{\text{KZ}}}{\frac{\lambda^2(L)}{\|\mathbf{b}_i\|^2\sin^{2}\theta_i}} \\
    &\leq \sup_{\mathbf{B}\in \mathcal{B}_{\text{KZ}}}{\frac{\lambda^2(L)}{\sum_{j=i}^n{r_j^2\|\mathbf{\hat{b}}_j\|^2}}}.\\
\end{split}
\end{equation}
Applying (\ref{A11}) yields
\[\rho_{i,\text{ZF}} \leq \sup_{\mathbf{B}\in \mathcal{B}_{\text{KZ}}} \lambda^2(L)\left[\|\mathbf{\hat{b}}_i\|^{-2}+\sum_{j=1}^{n-i} \left(\frac{1}{3}\right)^{2} \left(\frac{3}{2}\right)^{2j} \|\mathbf{\hat{b}}_{i+j}\|^{-2}\right].\]
Using the facts that $\lambda^2(L)=\|\mathbf{\hat{b}}_1\|^2$ and $\|\mathbf{\hat{b}}_{1}\|^2 \leq
\xi_{j}\|\mathbf{\hat{b}}_{j}\|^2$ for a KZ-reduced basis, we obtain (\ref{rhoiZFKZ2}).
\end{proof}


The maximum is attained when $i=1$, and
\begin{equation}\label{rhoZFKZ}
  \rho_{\text{ZF}} \leq 1+\sum_{j=1}^{n-1}{\left(\frac{1}{3}\right)^2
  \left(\frac{3}{2}\right)^{2j}\xi_{j+1}}.
\end{equation}
Using the upper bound (\ref{KZconstantBound2}), we obtain a more explicit form
\begin{equation}\label{rhoiZFKZ3}
\begin{split}
  \rho_{\text{ZF}} \leq \left(\frac{9}{4}\right)^{n-1} \cdot n^{1+\ln n}.
\end{split}
\end{equation}

\section{Dual Basis Reduction}

In this Section, we shall derive upper bounds on the proximity factors when the dual rather than
the primal basis is reduced.



\subsection{Dual LLL Reduction}

\subsubsection{SIC}
\begin{ther} For SIC with dual LLL reduction, the proximity factors are bounded by
\begin{equation}\label{}
\begin{split}
   \rho_{i,\text{SIC}} & \leq \beta^{i-1}.
\end{split}
\end{equation}
\end{ther}
\begin{proof}
A scrutiny into the derivation of $\rho_{\text{SIC}}$ with primal-basis reduction reveals that the
following condition is crucial:
\begin{equation}\label{EffectiveCondition}
  |\mu_{i,i-1}| \leq 1/2, \quad 1 < i \leq n.
\end{equation}
A basis satisfies the Lov\'asz condition and (\ref{EffectiveCondition}) is called an {\it
effectively} LLL-reduced basis \cite{Howgrave-Graham97}. Effective LLL reduction produces a basis
with the same set of Gram-Schmidt vectors. Accordingly, with SIC, it yields the same output as the
standard LLL reduction. Moreover, if a basis is effectively LLL-reduced, so is its dual basis
\cite{Howgrave-Graham97}. This is because the inverse of the coefficient matrix $\bm{\mu}$ admits
the form \cite{Howgrave-GrahamPrivate}
\[
\bm{\mu}^{-1}=
  \begin{bmatrix}
    {1} & {0} & {\cdots} & {0} & {0} \\
    {-\mu_{2,1}} & {1} & {\cdots} & {0} & {0} \\
    {\times} & {-\mu_{3,2}} & {1} & {\cdots} & {0} \\
    {} & {} & {} & {\ddots} & {} & {} \\
    {\times} & {\cdots} & {\times} & {-\mu_{n,n-1} }& {1} \
  \end{bmatrix}
\]
and the Lov\'asz condition can be shown to be satisfied for the dual basis.


Although other off-diagonal elements may have absolute values larger than $1/2$, it does not really
cause a problem, since we still have
\begin{equation}\label{DLLL-SIC-bound}
\begin{split}
   \rho_{i,\text{SIC}} &= \sup_{\mathbf{B}^*\in \mathcal{B}_{\text{LLL}}}{ \frac{\lambda^2(L)}{\|\mathbf{\hat{b}}_i\|^2}} \\
   & \leq \beta^{i-1}\sup_{\mathbf{B}^*\in \mathcal{B}_{\text{LLL}}}{\frac{\lambda^2(L)}{\|\mathbf{b}_1\|^2}}\\
   & \leq \beta^{i-1}.
\end{split}
\end{equation}
\end{proof}

Note that this is the same as (\ref{expbound}) for primal LLL reduction. Accordingly, we expect
similar performance for primal and dual LLL reductions.


\subsubsection{ZF}
\begin{ther} For ZF with dual LLL reduction, the proximity factors are bounded by
\begin{equation}\label{rhoiZF-dualLLL}
\begin{split}
    \rho_{i,\text{ZF}} &\leq \left(1 + \frac{\beta}{4}\cdot \frac{\beta^{n-i}-1}{\beta-1}\right)\beta^{i-1}.
\end{split}
\end{equation}
\end{ther}

\begin{proof}
The derivation of the bound on $\rho_{i,\text{ZF}}$ is analogous to that for primal LLL reduction.
We start by
\[\rho_{i,\text{ZF}} \leq \sup_{\mathbf{B}^*\in \mathcal{B}_{\text{LLL}}}{\frac{\lambda^2(L)}{\sum_{j=i}^n{r_j^2 \|\mathbf{\hat{b}}_j\|^2}}},\]
which holds for any bases. Incorporating the lower bound (\ref{generalization-dual}) in Appendix
\ref{appendix:dual}, we deduce
\begin{equation}\label{rho-DLLL-ZF}
    \rho_{i,\text{ZF}} \leq \sup_{\mathbf{B}^*\in \mathcal{B}_{\text{LLL}}} \lambda^2(L) \left(\|\mathbf{\hat{b}}_i\|^{-2} + \frac{1}{4}\sum_{j=1}^{n-i} \|\mathbf{\hat{b}}_{i+j}\|^{-2}\right).
\end{equation}
Besides, the primal basis is effectively LLL-reduced when the dual basis is LLL-reduced. Thus,
$\|\mathbf{\hat{b}}_i\|^2 \geq \beta^{-(i-1)}\|\mathbf{{b}}_1\|^2$. We obtain
\begin{equation}\label{rhoiZF-dualLLL2}
\begin{split}
    \rho_{i,\text{ZF}} &\leq \sup_{\mathbf{B}^*\in \mathcal{B}_{\text{LLL}}}{\frac{\lambda^2(L)}{\|\mathbf{{b}}_1\|^2}}\left(1 + \frac{1}{4}\sum_{j=1}^{n-i}
\beta^{j}\right)\beta^{i-1}\\
&= \left(1 + \frac{1}{4}\sum_{j=1}^{n-i} \beta^{j}\right)\beta^{i-1},
\end{split}
\end{equation}
which leads to (\ref{rhoiZF-dualLLL}).
\end{proof}

Since the right-hand side of (\ref{rhoiZF-dualLLL}) is an increasing function of $i$ for $\beta
\geq 4/3$, we have
\begin{equation}\label{rhoiZF-dualLLL2}
    \rho_{\text{ZF}} \leq \beta^{n-1}.
\end{equation}
We emphasize that the worst case bound (\ref{rhoiZF-dualLLL2}) is the same as dual LLL-SIC and is
close to that for SIC with primal LLL reduction.

\subsection{Dual KZ Reduction}

\subsubsection{SIC}
It is shown in \cite{lagarias} that if a reversed dual basis is KZ reduced, then
$\gamma_i\|\mathbf{\hat{b}}_i\| \geq \lambda(L)$. Accordingly, we have
\begin{ther} For SIC with dual KZ reduction, the proximity factors are bounded by
\begin{equation}\label{DKZSIC}
  \rho_{i,\text{SIC}} \leq \gamma_i^2.
\end{equation}
\end{ther}

\vspace{0.5cm}Although the Hermite constant is very likely to be an increasing function, it has
never been proved \cite{lagarias}. Hence, we bound the overall proximity factor for SIC by
\begin{equation}\label{DKZbound1}
  \rho_{\text{SIC}} \leq \max_{1\leq i\leq n}{\gamma_i^2}.
\end{equation}

Since $\gamma_n \leq 2n/3$, we have $\rho_{\text{SIC}} \leq (2n/3)^2 <n^2$. A better bound can be
obtained by employing Blichfeldt's bound on the Hermite constant:
\begin{equation}\label{DKZbound2}
  \rho_{\text{SIC}} \leq  \left(\frac{2}{\pi}\right)^2\Gamma^{4/n}\left(2+\frac{n}{2}\right).
\end{equation}

We claim that $\gamma_i^2$ is never greater than the bound on the KZ constant
(\ref{KZconstantBound}). We prove this by induction. When $i=1$, $\xi_1 = \gamma_1^2 = 1$, and the
claim is true. Suppose it is true when $k=i-1$, from which we deduce $\gamma_{i-1} \leq
\gamma_{i-1}^{1/(i-2)}...\gamma_2$. When $k=i$ we have
\[\frac{\gamma_i^2}{\gamma_i \gamma_{i}^{1/(i-1)} \gamma_{i-1}^{1/(i-2)}...\gamma_2} \leq \frac{\gamma_i}{\gamma_i^{1/(i-1)}\gamma_{i-1}} = \frac{\gamma_i^{(i-2)/(i-1)}}{\gamma_{i-1}}.\]
By Mordell's inequality $\gamma_i^{i-2} \leq \gamma_{i-1}^{i-1}$ \cite{Mordell}, this is not
greater than one. Thus the claim is true for $k=i$ as well.

Therefore, reducing the reversed dual basis again has a smaller bound on $\rho_{i,\text{SIC}}$ than
reducing the primal basis.

\subsubsection{ZF}

Substituting $\gamma_j\|\mathbf{\hat{b}}_j\| \geq \lambda(L)$ for dual KZ reduction \cite{lagarias}
into (\ref{rho-DLLL-ZF}), we arrive at
\begin{ther} For ZF with dual KZ reduction, the proximity factors are bounded by
\begin{equation}\label{rhoiZFKZ-dual}
  \rho_{i,\text{ZF}} \leq \gamma_i^2+\frac{1}{4}\sum_{j=i+1}^{n}\gamma_{j}^2.
\end{equation}
\end{ther}
\vspace{0.5cm}We claim that the right-hand side of (\ref{rhoiZFKZ-dual}) decreases with $i$. To see
this, it is sufficient to show $\gamma_i^2 \geq \frac{3}{4}\gamma^2_{i+1}$. We apply Mordell's
inequality together with the inequality $\gamma_n \leq (2/\sqrt{3})^{n-1}$, yielding \[\gamma_{i+1}
\leq \gamma_{i}^{i/(i-1)}=\gamma_i \gamma_i^{1/(i-1)} \leq \gamma_i \cdot 2/\sqrt{3}.\] Thus the
claim is proven.

Accordingly, the right-hand side of (\ref{rhoiZFKZ-dual}) attains the maximum when $i=1$, and
\begin{equation}\label{rhoZFKZ-dual}
  \rho_{\text{ZF}} \leq 1+\frac{1}{4}\sum_{j=2}^{n}\gamma_{j}^2.
\end{equation}
Using the bound $\gamma_j \leq 2j/3$ for $j \geq 2$, we can obtain
\begin{equation}\label{rhoiZFKZ2-dual}
  \rho_{\text{ZF}} \leq \frac{8}{9}+\frac{1}{9}\cdot\frac{n(n+1)(2n+1)}{6}.
\end{equation}
The right side of (\ref{rhoiZFKZ2-dual}) is never larger than $n^3$, and will be dominated by
$n^3/27$ as $n$ goes large.

\section{Numerical Results and Discussion}



To visualize the behaviors of proximity factors, we plot the best upper bounds in Table
\ref{tab:comparison}. In Fig. \ref{fig:Primal}, we show the upper bounds on the proximity factors
for KZ reduction and for LLL reduction with best performance (i.e., $\delta=1$). For KZ reduction,
we employ (\ref{rhoZFKZ}) and (\ref{KZconstantBound}) in conjunction with the bound
(\ref{Hermiteestimp}) on the Hermite constant and the exact value (\ref{Hermite8}) for $n \leq 8$.
It is clear that the bounds on proximity factors for KZ reduction are smaller than those for LLL
reduction. This is expected since KZ reduction is a stronger notion of lattice reduction.

\begin{figure}[t]

\centering\centerline{\epsfig{figure=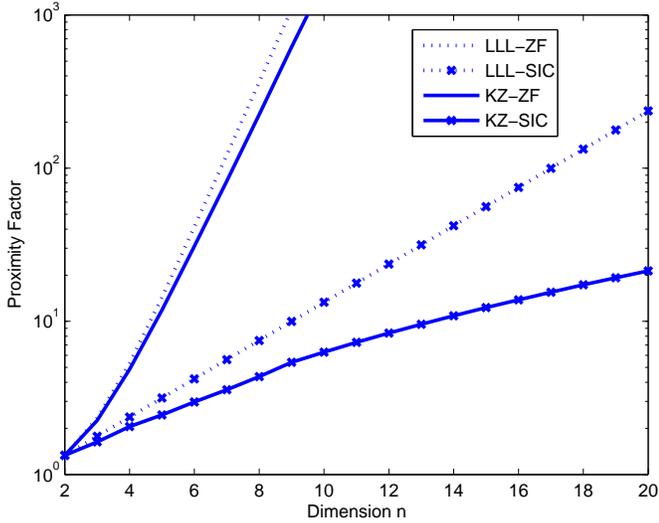,width=10cm}}

\vspace{-0.5cm}

\caption{Upper bounds on the proximity factors $\rho_{\text{ZF}}$ and $\rho_{\text{SIC}}$ for KZ
reduction and for LLL reduction with $\delta=1$.}

\vspace{-0.5cm}

\label{fig:Primal}
\end{figure}

Fig. \ref{fig:dual-LLL} shows the upper bounds on $\rho_{\text{ZF}}$ with primal and dual LLL
reduction. Obviously, the bound for reducing the dual basis is smaller for the purpose of ZF.

\begin{figure}[t]

\centering\centerline{\epsfig{figure=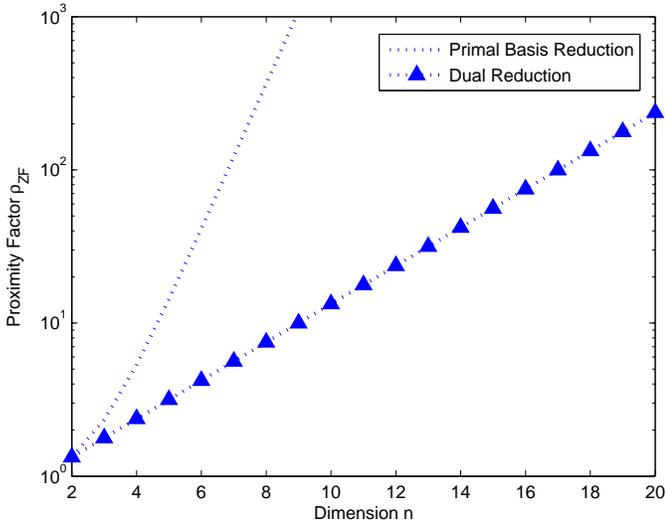,width=10cm}}


\caption{Upper bounds on the proximity factors $\rho_{\text{ZF}}$ for LLL reduction with
$\delta=1$.}

\vspace{-0.5cm}

\label{fig:dual-LLL}
\end{figure}

In Fig. \ref{fig:dual-KZ}, we show the upper bounds on the proximity factors for KZ reduction of
the primal as well as dual bases. Obviously, for ZF, the bound for reducing the dual basis is much
smaller. Meanwhile, we can see that, for SIC, the gain due to dual KZ reduction is not significant.
The reason is that as far as the Gram-Schmidt vectors are concerned, the primal and dual KZ
reduction are in fact not too much different, as also recognized in \cite{agrell}. Because the
product $\prod_{i=1}^{n} \|\mathbf{\hat{b}}_i\| = \det L$ is invariant under lattice reduction,
choosing short Gram-Schmidt vectors in the beginning will force long Gram-Schmidt vectors in the
end, and vice versa.

\begin{figure}[t]

\centering\centerline{\epsfig{figure=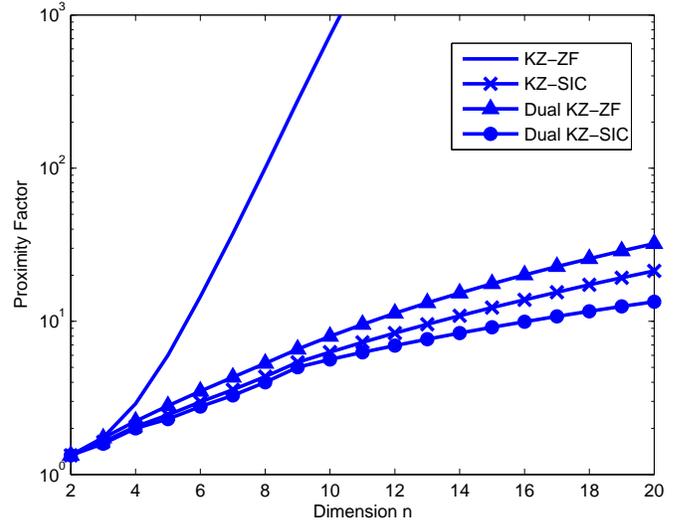,width=10cm}}

\vspace{-0.5cm}

\caption{Upper bounds on the proximity factors for KZ reduction of the primal and dual bases.}

\vspace{-0.5cm}

\label{fig:dual-KZ}
\end{figure}

Compared with primal basis reduction, the dual-basis reformulation is more natural in the following
senses.

For ZF, we have $\mathbf{B}^\dagger \mathbf{y}=\mathbf{x}+\mathbf{B}^\dagger\mathbf{n}$. Then, the
noise variance associated with the $i$-th output is proportional to $\|\mathbf{b}_i^*\|^2$.
Obviously, it makes sense to minimize $\|\mathbf{b}_i^*\|^2$.

Most existing reduction notions aim to find short vectors one after another; hence when applied to
the dual basis, they can be viewed as various greedy heuristics of finding the optimal bases for
decoding purposes. In particular, the dual Minkowski reduction can be viewed as a greedy algorithm
to approximate the ZF criterion, while the dual KZ reduction as a greedy algorithm to approximate
the SIC criterion, since they will successively find the shortest vector and the shortest
Gram-Schmidt vector for the dual basis, respectively.

To verify the benefit of dual reduction for LLL-ZF, we simulate the BER performance of an $8 \times
8$ MIMO system with i.i.d. complex standard normal channel coefficients. The complex basis matrix
$\mathbf{H}$ is converted into its real equivalent by
\[{\bf B} = \left[%
    \begin{array}{cc}
      \Re({\bf H}) & -\Im({\bf H}) \\
      \Im({\bf H}) & \Re ({\bf H}) \\
    \end{array}%
\right],\] and then $\mathbf{B}$ is reduced. The signal constellation at each antenna is 64QAM with
Gray mapping. With such a signal model, the SNR at each receive antenna is defined as $\text{SNR} =
n_T E_{x \in \text{64QAM}}[x^2]/\sigma^2$, where $n_T$ is the number of transmit antennas. Fig.
\ref{fig:BER_ZF} shows the simulated BER of ZF with primal LLL reduction and dual LLL reduction.
When the decoded lattice point falls outside of the signal boundary, we simply round it
componentwise back to the 64QAM alphabet. It is seen that the latter exhibits about 3 dB gain at
the BER of $10^{-5}$. Lattice decoding ignoring the boundary is also simulated, whose performance
is about 1.5 dB worse than ML decoding at the BER of $10^{-5}$.

\begin{figure}[t]

\centering\centerline{\epsfig{figure=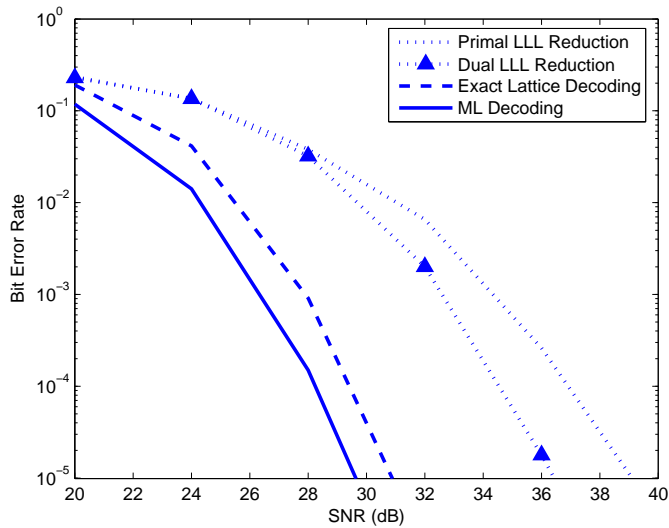,width=10cm}}

\vspace{-0.5cm}

\caption{BER of ZF with LLL reduction with $\delta = 3/4$ for 64QAM over an $8 \times 8$
complex-valued MIMO fading channel.}

\vspace{-0.5cm}

\label{fig:BER_ZF}
\end{figure}

A similar amount of gain is observed for ZF with primal and dual KZ reduction. In the simulation of
SIC, though, only a small difference in the BER is observed between primal and dual reduction. This
is not surprising given their close proximity factors (cf. Table \ref{tab:comparison}).




We assess the tightness of the bounds by means of numerical experimentation. Each time, a basis
matrix $\mathbf{B}$ is randomly generated and reduced, and then the minimum distance of the
decision region (for ZF and SIC, respectively) as well as the shortest vector is found. For each
value of $n$, the empirical proximity factor is obtained as the maximum over an ensemble of 10000
i.i.d. Gaussian matrices $\mathbf{B}$. While there is no guarantee that this maximum reaches the
worst-case bound, it should be a reasonable indicator of the theoretic proximity factor. In
particular, it can serve as a experimental lower bound on the theoretic proximity factor. Fig.
\ref{fig:PFLLL} shows the results for (primal-basis) LLL reduction with $\delta=3/4$. The numerical
results (in dB) still grow linearly, but at a lower slope. Fig. \ref{fig:PFLLL} also shows the real
SNR gap to ML decoding observed in computer simulation for $n = 8, 12, 16$. To estimate the real
SNR gap, the BER performance is simulated, and curves like those in Fig. \ref{fig:BER_ZF} are
obtained. It can be seen that the real SNR gap is even smaller. This behavior is consistent with
the common belief that in practice LLL reduction performs better than the theoretic upper bounds.
Fig. \ref{fig:PFKZ} demonstrates the numerical results for KZ reduction. The similar trend can be
seen.

\begin{figure}[t]

\centering\centerline{\epsfig{figure=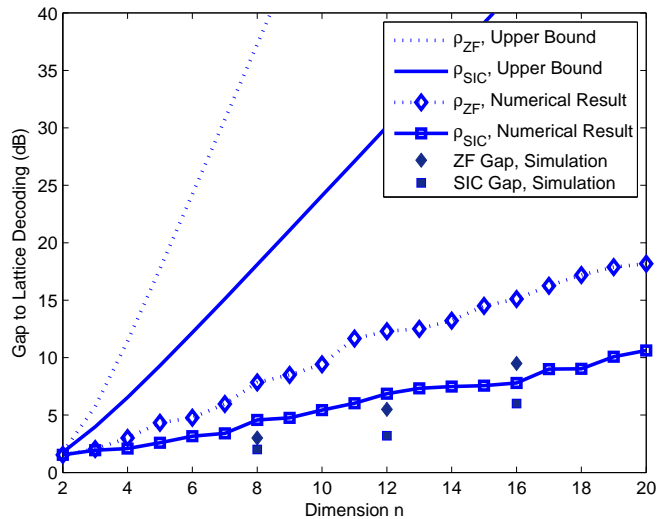,width=10cm}}

\caption{Comparison of the theoretic upper bounds and numerical results of the proximity factors
for LLL reduction ($\delta=3/4$).}


\label{fig:PFLLL}
\end{figure}

\begin{figure}[t]

\centering\centerline{\epsfig{figure=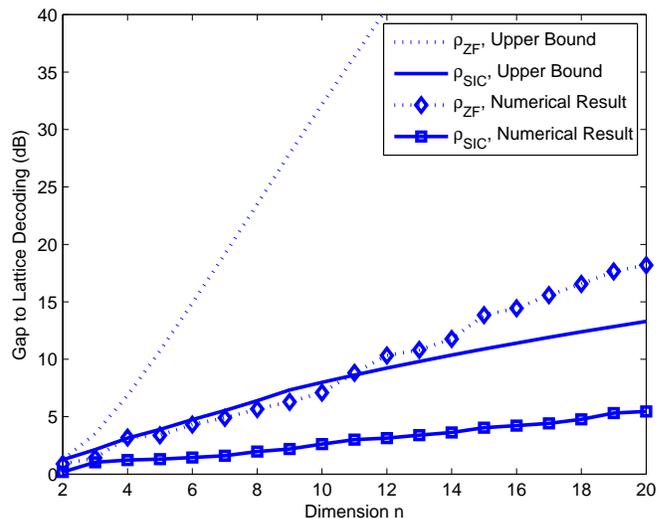,width=10cm}}

\caption{Comparison of the theoretic upper bounds and numerical results of the proximity factors
for KZ reduction.}


\label{fig:PFKZ}
\end{figure}

In summary, we have shown that lattice reduction-aided decoding is in proximity to (infinite)
lattice decoding. We have derived analytic bounds on proximity factors that quantify the worst-case
gap. The same diversity order as that of lattice decoding in MIMO fading channels comes as a direct
consequence. According to the analysis, dual lattice reduction could give smaller upper bounds,
especially in ZF where performance gain has been observed in computer simulation.

Finally, we point out several open issues of lattice reduction-aided decoding.

Firstly, there is room to tighten the bounds derived in this paper. In Table \ref{tab:comparison},
the bounds are tight for Gauss reduction and for KZ-SIC with primal reduction. The bound is very
likely to be tight for KZ-SIC with dual reduction; it will be tight as long as the Hermite constant
is an increasing function which is very likely to be true. The bound is quite good for KZ-ZF with
dual reduction (cf. Fig. \ref{fig:dual-KZ}). However, the bounds are not tight for LLL reduction
and KZ-ZF with primal reduction. For LLL reduction, we have applied the trivial bound
$\lambda^2(L)/\|\mathbf{b}_i\|^2 \leq 1$ or $\lambda^2(L)/\|\mathbf{b}_1\|^2 \leq 1$. But this is
likely to loosen the bound, since $\beta^{-(n-1)} \leq \lambda^2_i(L)/\|\mathbf{b}_i\|^2$ and
$\beta^{-(n-1)} \leq \lambda^2(L)/\|\mathbf{b}_1\|^2$ for LLL reduction \cite{LLL}.

Secondly, we emphasize that the worst-case bounds should be carefully interpreted. The average-case
performance may be a more meaningful measure if the basis $\mathbf{B}$ is random. A smaller
proximity factor does not necessarily guarantee better average performance. For example, KZ
reduction must have smaller approximation factors than LLL reduction, but their BER performance
appears to be very close in MIMO fading channels \cite{Windpassinger2}. Hence an average-case
analysis will complement the worse-case analysis. In simulations of LLL reduction, we observed that
the SNR gap (in dB) indeed widens linearly with the dimension, but at a lower slope. This behavior
is consistent with the common belief that in practice LLL reduction performs better than the
theoretic upper bounds. Recently, Nguyen and Stehl\'e showed that the average behavior of LLL
reduction is still exponential, but at a slower rate \cite{Nguyen1}.

Thirdly, the analysis of dual reduction naturally raises the issue of simultaneously reducing the
primal and dual bases. Methods of simultaneous reduction exist, notably those of Seysen
\cite{Seysen1} and Howgrave-Graham \cite{Howgrave-GrahamIso}.

\appendices


\section{The Angle $\theta_i$ of a Basis}\label{appendix:theta}

%
%
%
%

We derive a formula for angel $\theta_i$ of a basis $\mathbf{B}$, which is used to compute the
distance of ZF. Following Babai \cite{Babai}, let $\mathbf{a}=\sum_{j \neq i}\alpha_j\mathbf{b}_j$
($\alpha_j \in \mathbb{R}$ for $j \neq i$) be an arbitrary vector on the hyperplane spanned by
vectors $\mathbf{b}_1, \cdots, \mathbf{b}_{i-1}, \mathbf{b}_{i+1}, \cdots, \mathbf{b}_n$. The
vector $\mathbf{z} = \mathbf{a}-\mathbf{b}_i$ will attain its minimum length if and only if it is
perpendicular to the hyperplane. We express $\mathbf{z}$ in terms of the Gram-Schmidt vectors as
$\mathbf{z} = \sum_{j=1}^n{r_j \mathbf{\hat{b}}_j}$. Set $\alpha_i = -1$. Then we have the
expression \cite{LLL}
\[r_j = \alpha_j + \sum_{t=j+1}^n{\alpha_t \mu_{t,j}}.\]
We can see
\begin{equation}\label{AngleDistance}
    \|\mathbf{b}_i\|^2 \sin^2 \theta_i = \min \|\mathbf{z}\|^2 = \min
\sum_{j=1}^n{r_j^2\|\mathbf{\hat{b}}_j\|^2}
\end{equation}
where the minimum is taken over the coefficients $\alpha_j$. It is not difficult to see that $r_j =
0$ for $j<i$ if $\mathbf{z}$ is orthogonal to the hyperplane. Accordingly, we may exclude the $j <
i$ terms, to yield
\begin{equation}\label{BabaiAngle}
    \|\mathbf{b}_i\|^2 \sin^2 \theta_i = \min \sum_{j=i}^n{r_j^2\|\mathbf{\hat{b}}_j\|^2}.
\end{equation}
This was also done by Babai \cite{Babai}. Here we have explicitly shown that excluding the $j < i$
terms does not weaken the bound.

For notational convenience, define ${\bm \alpha} = [\alpha_i,\cdots, \alpha_n]^T$, $\mathbf{r} =
[r_i,\cdots, r_n]^T$, $\mathbf{M}$ be $(n-i+1) \times (n-i+1)$ bottom-right submatrix of
$\bm{\mu}$, ${\bm \Lambda}$ be the diagonal matrix $\text{diag}(\|\mathbf{\hat{b}}_i\|^2,
\|\mathbf{\hat{b}}_{i+1}\|^2, \cdots, \|\mathbf{\hat{b}}_n\|^2)$, and $\mathbf{A}=\mathbf{M}
\bm{\Lambda} \mathbf{M}^T$. We have the expressions $\mathbf{r} = \mathbf{M}^T\bm{\alpha}$ and
\[\sum_{j=i}^n{r_j^2 \|\mathbf{\hat{b}}_j\|^2} = \mathbf{r}^T\bm{\Lambda}\mathbf{r} = \bm{\alpha}^T \mathbf{A} \bm{\alpha}.\] We want to find the minimum
of $\bm{\alpha}^T \mathbf{A} \bm{\alpha}$ subject to the constraint $\alpha_i = -1$. This can be
achieved by using the Lagrangian multiplier method. Define the objective function
\[f = \bm{\alpha}^T \mathbf{A} \bm{\alpha} + s(\alpha_i+1).\]
Nulling the partial derivative of $f$ with respect to $\bm{\alpha}$, we have
\[2\mathbf{A} \bm{\alpha} + s\mathbf{e}_1 = 0\]
where the unit vector $\mathbf{e}_1$ has all zero elements except that the first is one. From this
we obtain $\bm{\alpha} = -\frac{s}{2}\mathbf{A}^{-1}\mathbf{e}_1$. Furthermore, since we know
$\alpha_i = -\frac{s}{2}(\mathbf{A}^{-1})_{1,1} = -1$, we have $s = 2/(\mathbf{A}^{-1})_{1,1}$.
Therefore,
\begin{equation}\label{}
    f_{min} = \frac{s^2}{4} \mathbf{e}_i^T\mathbf{A}^{-1}\mathbf{e}_i = \frac{s^2}{4}
    (\mathbf{A}^{-1})_{1,1} = \frac{1}{(\mathbf{A}^{-1})_{1,1}},
\end{equation}
and from (\ref{AngleDistance}) we deduce
\begin{equation}\label{sintheta}
     \sin^2 \theta_i= \frac{1}{\|\mathbf{b}_i\|^2 (\mathbf{A}^{-1})_{1,1}}.
\end{equation}
We have expressed $ \theta_i$ in terms of $\mathbf{b}_i$ and the Gram-Schmidt orthogonalization of
$\mathbf{B}$ . Given $\mathbf{B}$, the angle $\theta_i$ can easily be calculated.

%
%
\section{Proof of Lemma \ref{lem:lemma1}}\label{appendix:Newgamma}

When deriving a lower bound on $\theta_i$, Babai loosened the bound in several steps. It is natural
to ask if his bound can be improved. In Appendix~\ref{appendix:theta}, we have seen that excluding
the $j<i$ terms from (\ref{AngleDistance}) does not weaken the bound. Babai showed \cite{Babai}
\begin{equation}\label{Babaigamma}
    \sum_{j=i}^n{r_j^2 \|\mathbf{\hat{b}}_j\|^2} \geq (2/9)^{n-i}\|\mathbf{\hat{b}}_i\|^2.
\end{equation}
Here we derive a better lower bound by using a different approach. Basically, we examine the
maximum of the first element of matrix $\mathbf{A}^{-1} =
\mathbf{M}^{-T}\bm{\Lambda}^{-1}\mathbf{M}^{-1}$ for a size-reduced (lower-triangular) matrix
$\mathbf{M}$. To do so, we need the following lemma:

\begin{lem}\label{lem:lemma4}
The absolute value of the ($i,j$)-th entry of the inverse $\mathbf{M}^{-1}$ of an $l \times l$
size-reduced matrix $\mathbf{M}$ is not greater than
\begin{equation}\label{muinverse}
  \frac{1}{3}\cdot \left(\frac{3}{2}\right)^{i-j}, \quad 1 \leq j < i \leq l,
\end{equation}
and this is achieved when all off-diagonal elements of $\mathbf{M}$ are equal to $-1/2$.
\end{lem}

\begin{proof}
We prove it by induction on the dimension $l$. Note that $\mathbf{M}^{-1}$ is lower triangular,
with unit diagonal elements.

Lemma \ref{lem:lemma4} is obviously true when $l=2$, since the off-diagonal element of
$\mathbf{M}^{-1}$ is $-m_{2,1}$, and the maximum $1/2$ is achieved when $m_{2,1}=-1/2$.

Suppose Lemma \ref{lem:lemma4} is true for $l = k-1$. When $l=k$, we partition the matrix
$\mathbf{M}$ into the form
\[ \mathbf{M}=
  \begin{bmatrix}
    {\mathbf{M}_{k-1}} & {\mathbf{0}} \\
    {\mathbf{h}^T} & {1}
  \end{bmatrix}.
\]
Using the formula for the inverse of a partitioned matrix \cite{Horn}, $\mathbf{M}^{-1}$ can be
expressed as
\[ \mathbf{M}^{-1}=
  \begin{bmatrix}
    {\mathbf{M}_{k-1}^{-1}} & {\mathbf{0}} \\
    {-\mathbf{h}^T\mathbf{M}_{k-1}^{-1}} & {1}
  \end{bmatrix}.
\]
Here, the $(i,j)$-th element of $\mathbf{M}_{k-1}^{-1}$ is given by (\ref{muinverse}) for $1 \leq j
< i \leq k-1$. Since they are all positive, each element of $-\mathbf{h}^T\mathbf{M}_{k-1}^{-1}$ is
maximized when $\mathbf{h}$ is an all $-1/2$ vector.

Consider the first element of $-\mathbf{h}^T\mathbf{M}_{k-1}^{-1}$, given by
\[\frac{1}{2}\left[1+\sum_{j=1}^{k-2}\frac{1}{3}\cdot \left(\frac{3}{2}\right)^{j}\right] = \frac{1}{3}\cdot \left(\frac{3}{2}\right)^{k-1},\]
which verifies (\ref{muinverse}). The other elements of $-\mathbf{h}^T\mathbf{M}_{k-1}^{-1}$ can be
verified in the same way.
\end{proof}

Denote by $\mathbf{p}$ the first column of $\mathbf{M}^{-1}$. Then we have
\begin{equation}\label{A11}
\begin{split}
    (\mathbf{A}^{-1})_{1,1} &= \mathbf{p}^T \bm{\Lambda}^{-1} \mathbf{p}\\
&\leq \|\mathbf{\hat{b}}_i\|^{-2}+\sum_{j=1}^{n-i} \left(\frac{1}{3}\right)^{2} \left(\frac{3}{2}\right)^{2j} \|\mathbf{\hat{b}}_{i+j}\|^{-2}.\\
\end{split}
\end{equation}
Note that the only condition for (\ref{A11}) is that the basis is size-reduced, i.e., $|\mu_{i,j}|
\leq 1/2$ for $i < j$. The choice of the parameter $\delta$ would have no effect on this condition.

For an LLL-reduced basis, $\|\mathbf{\hat{b}}_{i+j}\|^{2} \geq
\beta^{-j}\|\mathbf{\hat{b}}_i\|^{2}$. Hence,
\begin{equation}\label{A11LLL}
\begin{split}
    (\mathbf{A}^{-1})_{1,1} &\leq \left[1+\sum_{j=1}^{n-i} \left(\frac{1}{3}\right)^{2} \left(\frac{3}{2}\right)^{2j} \beta^j\right]\|\mathbf{\hat{b}}_i\|^{-2}\\
& = \left[\frac{\beta}{9\beta-4}\left(\frac{9\beta}{4}\right)^{n-i} +
\frac{8\beta-4}{9\beta-4}\right]\|\mathbf{\hat{b}}_i\|^{-2},
\end{split}
\end{equation}
and correspondingly,
\begin{equation}\label{Newgamma}
    \sum_{j=i}^n{r_j^2 \|\mathbf{\hat{b}}_j\|^2} \geq \left[\frac{\beta}{9\beta-4}\left(\frac{9\beta}{4}\right)^{n-i} + \frac{8\beta-4}{9\beta-4}\right]^{-1}\|\mathbf{\hat{b}}_i\|^2.
\end{equation}
When $\beta=2$, the new bound is asymptotically tighter by a factor of 7 than Babai's lower bound
(\ref{Newgamma}).

Substituting (\ref{vectorlength}) and (\ref{A11LLL}) into (\ref{sintheta}), we obtain Lemma
\ref{lem:lemma1}. Note that we have also incorporated the better estimate (\ref{vectorlength}) of
the length of $\mathbf{b}_i$.

\section{Lower Bound for Dual Size Reduction}\label{appendix:dual}
We derive the counterparts of (\ref{A11}) and (\ref{Newgamma}) when the reversed dual basis is
size-reduced. The notation in Appendix \ref{appendix:Newgamma} is followed. The difference is that,
by the following lemma, all off-diagonal elements of $\mathbf{M}^{-1}$ lie in interval
$[-1/2,1/2]$.

\begin{lem}\label{lem:lemma3}
If the dual basis is size-reduced, then the off-diagonal elements of $\bm{\mu}^{-1}$ lie in the
interval $[-1/2, 1/2]$.
\end{lem}

\begin{proof}
By definition, the off-diagonal elements of $\bm{\mu}^*$ lie in $[-1/2, 1/2]$ if the reversed dual
basis $\mathbf{B}^*$ is size-reduced. By Proposition \ref{prop1}, $\bm{\mu}^{-1} =
\mathbf{J}(\bm{\mu}^{*})^T\mathbf{J}$. This proves Lemma \ref{lem:lemma3}.
\end{proof}

Consequently,
\begin{equation}\label{A11-dual}
\begin{split}
    (\mathbf{A}^{-1})_{1,1} &= \mathbf{p}^T \bm{\Lambda}^{-1} \mathbf{p}\\
&\leq \|\mathbf{\hat{b}}_i\|^{-2} + \frac{1}{4}\sum_{j=1}^{n-i} \|\mathbf{\hat{b}}_{i+j}\|^{-2},
\end{split}
\end{equation}
and
\begin{equation}\label{generalization-dual}
    \sum_{j=i}^n{r_j^2 \|\mathbf{\hat{b}}_j\|^{2}} \geq \left(\|\mathbf{\hat{b}}_i\|^{-2} + \frac{1}{4}\sum_{j=1}^{n-i} \|\mathbf{\hat{b}}_{i+j}\|^{-2}\right)^{-1}.
\end{equation}

In contrast, when the primal basis is size-reduced, the off-diagonal elements of $\mathbf{M}^{-1}$
can only be bounded by (\ref{muinverse}), which can be much larger than $1/2$. This analysis
clearly shows that it is the dual basis, rather than the primal basis, that should be reduced for
the purpose of ZF.

\section*{Acknowledgment}
The author would like to thank W. H. Mow and N. Howgrave-Graham for helpful discussions. Thanks are
also due to the anonymous reviewers for useful comments that led to improved presentation of this
paper.


\footnotesize
\bibliographystyle{IEEEtran}
\bibliography{IEEEabrv,lingbib}

%





\end{document}